\documentclass[12pt]{article}
\usepackage{amsfonts,amssymb,amsmath}
\usepackage{epsfig}
\usepackage{amsthm}
\usepackage{enumerate}
\makeatletter
\@addtoreset{equation}{section}

\makeatother

\usepackage[usenames,dvipsnames]{color}

\newcommand\ZZ{\mathbb Z}
\newcommand\RR{\mathbb R}
\newcommand\CC{\mathbb C}

\newcommand\beq{\begin{equation}}
\newcommand\eeq{\end{equation}}

\newtheorem{theorem}{Theorem}
\newtheorem{lemma}{Lemma}

\newtheorem{remark}{Remark}
\newtheorem{proposition}{Proposition}

\makeatletter
\@addtoreset{theorem}{section}
\@addtoreset{lemma}{section}
\@addtoreset{definition}{section}
\@addtoreset{remark}{section}
\@addtoreset{proposition}{section}
\@addtoreset{statement}{section}
\makeatother

\begin{document}

\title{Transmission eigenvalues for multipoint scatterers\thanks{The work was partially supported by a joint grant of
the Russian Foundation for Basic Research and CNRS (projects no. RFBR 20-51-1500/PRC no. 2795 CNRS/RFBR).}}
\author{P.G. Grinevich
  \thanks{Steklov Mathematical Institute of RAS, 8 Gubkina St. Moscow,
119991, Russia; 
Landau Institute of Theoretical Physics, pr. Akademika Semenova 1a,
Chernogolovka, Moscow region, 142432, Russia;
Moscow State University, Moscow, Russia. e-mail: pgg@landau.ac.ru} 
\and R.G. Novikov\thanks
{CMAP, CNRS, \'Ecole Polytechnique,
Institut Polytechnique de Paris,
Palaiseau, France.
e-mail: novikov@cmap.polytechnique.fr}}
\date{}

\maketitle

\begin{abstract}
  We study the transmission eigenvalues for the multipoint scatterers of the Bethe-Peierls-Fermi-Zeldovich-Beresin-Faddeev type in dimensions $d=2$ and $d=3$. We show that for these scatterers: 1) each positive energy $E$ is a transmission eigenvalue (in the strong sense) of infinite multiplicity; 2) each complex $E$ is an interior transmission eigenvalue of infinite multiplicity. The case of dimension $d=1$ is also discussed.
\end{abstract}



\section{Introduction}
Studies of transmission eigenvalues for Helmhotz-type equations have a long history, and many publications are devoted to this problem. In fact, the property that a spectral parameter $E$ is a transmission eigenvalue, is a weakened version of invisibility (transparency) for this $E$. In the case of the stationary Schr\"odinger equation the spectral parameter $E$ is interpreted as the energy.

In connection with transparency for the Schr\"odinger equation,  see \cite{CS,GRN3,G} and references therein. Historically, these studies go back to \cite{Red,New}. In connection with transmission eigenvalues for Helmhotz-type equations, see \cite{CH,NgNg,CN} and references therein. Historically, these studies go back to \cite{Kir,CM}.

In particular, it is known that for sufficiently regular compactly supported scatterers transmission eigenvalues are discrete and have finite multiplicity (see \cite{RS,NgNg,CN} for precise results). On the other hand, in \cite{GRN3} we constructed two-dimensional real-valued potentials from the Schwartz class, which are transparent at one positive energy. As a corollary, this energy is a transmission eigenvalue of infinite multiplicity in the strong sense, i.e., the number $1$ is an eigenvalue of infinite multiplicity for the scattering operator at this fixed energy! However, the potentials of \cite{GRN3} are not compactly supported.

In the present work we consider the stationary Schr\"odiger equation 
\beq
\label{eq:1}
-\Delta\psi + v(x)\psi = E\psi,  \ \ x\in\RR^d, \ \ d=1,2,3,
\eeq
with multipoint potential (scatterer) 
\beq
\label{eq:1.2}
v(x)=\sum\limits_{j=1}^{n} \varepsilon_j\delta(x-y_j).
\eeq

It is well-known that point scatterers $\varepsilon\delta(x)$ are well-defined only in dimensions $d=1,2,3$. If $d=1$, then $\delta(x)$ denotes the standard Dirac $\delta$-function. If $d=2$ or $d=3$, then $\varepsilon\delta(x)$ denotes a ``renormalization'' of the  $\delta$-function (depending on a parameter). These ``renormalized'' $\delta$-functions for $d=2,3$ are known as Bethe-Peierls-Fermi-Zeldovich-Beresin-Faddeev-type point scatterers. Exact definitions for the multipoint scatterers in dimensions $d=2,3$ can be found in \cite{AGHH,GRN4,GRN5}. Historically, mathematical definitions of point scatterers in dimension $d=3$ go back to  \cite{BP,BF}.

We also consider the equation 
\beq
\label{eq:1.1}
-\Delta\psi = E\psi,  \ \ x\in\RR^d, \ \ d=1,2,3.
\eeq

Let  ${\cal U}_j$ be an open non-empty neighbourhood of $y_j$ such that $y_{j'}\not\in {\cal U}_j$ if $j'\ne j$. 
Developing the approach, used in \cite{BP} for $d=3$, a local solution $\psi$ of (\ref{eq:1}), (\ref{eq:1.2}) in $U_j$ can be defined as a function such that:
\begin{enumerate}[(i)]
\item $\psi(x)$ satisfies   (\ref{eq:1.1}) in ${\cal U}_j\backslash y_j$;
\item If $d=1$, then $\psi(x)$ is continuous at $x=y_j$, and its first derivative has a jump
\beq\label{eq:1.4}  
-\alpha_j\left[\psi'(y_j+0)- \psi'(y_j-0)\right] = \psi(y_j);
\eeq
If $d=2$, then
\beq\label{eq:1.5}  
\psi(x) = \psi_{j,-1}\ln|x-y_j| + \psi_{j,0} + O(|x-y_j|) \ \ \mbox{as} \ \ x\rightarrow y_j,
\eeq
and
\beq\label{eq:1.6}  
\left[-2 \pi \alpha_j -\ln 2 + \gamma\right] \psi_{j,-1} = \psi_{j,0},
\eeq
where $\gamma=0.577\ldots$ is the Euler's constant.\\
If $d=3$, then
\beq\label{eq:1.7}  
\psi(x) = \frac{\psi_{j,-1}}{|x-y_j|}+ \psi_{j,0} + O(|x-y_j|)\ \ \mbox{as} \ \ x\rightarrow y_j,
\eeq
and
\beq\label{eq:1.8} 
4\pi \alpha_j \psi_{j,-1} = \psi_{j,0}.
\eeq
\end{enumerate}

Here, we assume that $\alpha_j\in\RR\cup\infty$. In addition, for each  $j=1,\ldots,n$, the strength $\varepsilon_j$ of the point scatterer $\varepsilon_j\delta(x-y_j)$ in (\ref{eq:1.2}) is encoded by $\alpha_j$, see, for example, \cite{AGHH,GRN5}. If $\alpha_j=\infty$, then $\epsilon_j=0$ (and the corresponding single point scatterer vanishes). If $d=1$, then the renormalized $\delta$-function coincides with the standard one, and we have that $\epsilon_j=-1/\alpha_j$. 

In the present work, we use, in particular, that the scattering functions for our multipoint scatterers are given by formulas (\ref{eq:2.6})--(\ref{eq:2.12}) recalled in Section~\ref{sect:2}.

In the present work we show that for the Schr\"odinger equation (\ref{eq:1}) with multipoint potentials (\ref{eq:1.2}), for $d=2,3$, each positive energy $E$ is a transmission eigenvalue of infinite multiplicity in the strong sense, and each complex $E$ is an interior transmission eigenvalue of infinite multiplicity. The related definitions of transmission eigenvalues are recalled in  Section~\ref{sect:2}. One can see that these multipoint potentials (\ref{eq:1.2}) are compactly supported, but they are singular. We also give proper analogs of these results for $d=1$ for single-point potentials.

In Section~\ref{sect:2} we also recall some basic facts from the scattering theory. In connection with inverse scattering for the multipoint scatterers of Bethe-Peierls-Fermi-Zeldovich-Beresin-Faddeev type see, for example, \cite{AN,BBMR,DR,GRN6,Nov}  and references therein.

The main results of the present work are presented in details in Section~\ref{sect:3}. The proofs are given in Section~\ref{sect:4}.

\section{Preliminaries}\label{sect:2}

For equation (\ref{eq:1}) with multipoint potentials as in (\ref{eq:1.2}) we consider the scattering eigenfunctions $\psi^+$ specified by the following asymptotics as $|x|\rightarrow\infty$:
\begin{equation}
\label{eq:3}
\psi^+= e^{ikx}+ f^+ \left(k,|k|\frac{x}{|x|}\right) \frac{e^{i|k||x|}}{|x|^{(d-1)/2}} +O\left(\frac{1}{|x|^{(d+1)/2}} \right), 
\end{equation}
$k\in\RR^d$, $k^2=E>0$, where $f^+=f^+(k,l)$, $k,l\in\RR^d$, $k^2=l^2=E$, is a priori unknown. The function $f^+$
arising in (\ref{eq:3}) is the scattering amplitude, or the far-field pattern. It is also convenient to present $f^+$ as follows:
\begin{align}
&f^+(k, l) = c(d, |k|)f(k, l), \label{f1}\\
&c(d, |k|) = -\pi i (-2\pi i)^{(d-1)/2}|k|^{(d-3)/2}, \, \text{ for } \sqrt{-2\pi i} = \sqrt{2\pi} e^{-i\pi/4}\label{c}.
\end{align}
 
It is important that for the multipoint scatterers the scattering eigenfunctions and scattering amplitudes are calculated explicitly (see, for example, \cite{AGHH,GRN4,GRN5}).

Let
\begin{equation}
G^+(x, E) := -(2\pi)^{-d} \int_{\mathbb{R}^d} \frac{e^{i\xi x}d\xi}{\xi^2-E -i\cdot 0} 
\end{equation}
where $x\in\RR^d, \ \ E\in\RR, \ \ E>0$. We recall that $G^+(x, E)$ is the Green function with the Sommerfeld radiation condition for the operator $\Delta+E$. Note also that 
\begin{align}
\begin{aligned}\label{G+}
&G^+(x, E) = \frac{e^{i|k||x|}}{2i|k|}, \, d=1,\\
&G^+(x, E) = -\frac{i}{4} H^1_0(|x||k|), \\
&\hphantom{aaa}  \mbox{where} \ H^1_0 \ \mbox{is the Hankel function of the first type,} \ d=2, \ \mbox{and}\\
&G^+(x, E) = \frac{1}{2\pi}\left[\ln|x| + \ln|k|-\ln 2 + \gamma - \frac{\pi i}{2}  \right] + O(|x|^2|\ln|x||), \ \ \mbox{as} \ \ |x|\rightarrow 0,\\
&\hphantom{aaa}  \mbox{where} \ \gamma=0.577\ldots \ \mbox{ is the Euler's constant}, \ \ \ \  {d=2,} \\
&G^+(x, k) = -\frac{e^{i|k||x|}}{4\pi|x|},\,\, d=3,
\end{aligned}
\end{align}
where $|k|=\sqrt{E}>0$.

Then the following formulas hold (see \cite{AGHH,GRN4,GRN5}):
\begin{equation}
\label{eq:2.6}
\psi^+(x,k)=e^{ikx}  +  \sum\limits_{j=1}^n q_j(k) G^+(x-y_j,|k|^2),
\eeq
\beq
\label{eq:2.7}
f(k,l) = \frac{1}{(2\pi)^d}  \sum\limits_{j=1}^n q_{j}(k) 
e^{-il y_j}.
\eeq
Here
$q(k)=(q_{1}(k),\ldots,q_{n}(k))^t$ satisfies the system of linear equations 
\begin{equation}
\label{eq:2.10}
A(|k|) q(k) = b(k),
\eeq
where $A$ is the $n\times n$ matrix with the elements
\begin{align}
&A_{j,j}(|k|) &= &\ \ \alpha_j - i(4\pi)^{-1}|k| , \ \ &d=3,\nonumber\\ 
\label{eq:2.11}
  &A_{j,j}(|k|) &= &\ \ \alpha_j - (4\pi)^{-1}(\pi i -2 \ln(|k|)), \ \ &d=2,\\
 &A_{j,j}(|k|) &= &\ \ \alpha_j + (2 i |k|)^{-1}, \ \ &d=1,\nonumber\\  
&A_{j,j'}(|k|) &= &\ \ G^{+}(y_j-y_{j'},|k|^2), \ \ &m\ne j,\nonumber
\end{align}
and $b(k)=(b_{1}(k),\ldots,b_{n}(k))^t$ is defined by
\begin{equation}
\label{eq:2.12}
b_{j}(k) = -e^{iky_j}.
\end{equation}  

In connection with (\ref{eq:1.4})-- (\ref{eq:1.8}) and (\ref{eq:2.6}) note also that the following formulas hold:
 \begin{align}
& &-\frac{1}{4\pi} q_j &= & &\psi_{j,-1}, \ \ &d=3,\nonumber\\ 
\label{eq:2.12.1}
&  &\frac{1}{2\pi} q_j &= & &\psi_{j,-1}, \ \ &d=2,\\
& &q_j                &= & &\psi'(y_j+0)-   \psi'(y_j-0),  \ \ &d=1.\nonumber
\end{align}

As we mentioned in Introduction, for each  $j=1,\ldots,n$, the strength $\varepsilon_j$ of the point scatterer $\varepsilon_j\delta(x-y_j)$ in (\ref{eq:1.2}) is encoded by a real parameter $\alpha_j$; see, also, for example, \cite{AGHH,GRN5}. 

The scattering operator $\hat S=\hat S_E$ at fixed energy $E=|k|^2$ can be defined as follows:
\beq
\label{eq:2.13}
(\hat S_E u)(\theta) =  u(\theta) -i\pi |k|^{d-2}  \int_{{\mathbb S}^{d-1}} f(|k|\theta', |k|\theta) u(\theta') d \theta',
\eeq
where ${\mathbb S}^{d-1}$  is the unit sphere in $\RR^d$, $\theta$, $\theta'\in {\mathbb S}^{d-1}$, $d\theta'$ denotes the standard volume element at ${\mathbb S}^{d-1}$, see, for example, \cite{LL}, Chapter XVII, \S~125, see also \cite{CH2}. 

Consider the equation
\begin{equation}
\label{eq:tr_eig}
(\hat S_E u) (\theta) = u (\theta), \ \  \theta\in {\mathbb S}^{d-1},
\end{equation}  
where $u\in L^2({\mathbb S}^{d-1})$. We say that energy $E$ is a \textbf{transmission eigenvalue in the strong sense} if equation (\ref{eq:tr_eig}) has a non-trivial solution.
Dimension of space of these solutions is the \textbf{multiplicity} of this transmission eigenvalue.

Assume that
\begin{equation}
\label{eq:supp}
\mbox{supp} \ \ v \subset \cal D.
\end{equation}
where $\cal D$ is a connected bounded domain in $\RR^d$ with $C^2$ boundary, where $\RR^d \backslash \overline{\cal D}$ is also connected.
Let
\begin{equation}
\label{eq:tr_eig1}
\psi(x) =  \int_{{\mathbb S}^{d-1}} \psi^+(x, |k|\theta') u(\theta') d \theta', 
\end{equation}
\begin{equation}
\label{eq:tr_eig2}
\phi(x) =  \int_{{\mathbb S}^{d-1}} e^{i|k|\theta' x} u(\theta') d \theta',
\end{equation}
where $u$ satisfies (\ref{eq:tr_eig}), $|k|^2=E$. Then
\begin{equation}
\label{eq:schr1}
-\Delta \psi(x) + v(x) \psi(x) = E  \psi(x), \ \ x\in{\cal D},
\end{equation}
\begin{equation}
\label{eq:schr2}
-\Delta \phi(x)  = E  \phi(x), \ \ x\in{\cal D},
\end{equation}
and
\begin{equation}
\label{eq:tr_eig3}
\psi(x)\equiv \phi(x), \ \ \frac{\partial}{\partial\nu}\psi(x)\equiv \frac{\partial}{\partial\nu}\phi(x) \ \ \mbox{for all} \ \ x\in\partial{\cal D},
\end{equation}
where $\frac{\partial}{\partial\nu}$ denotes the normal derivative, see, for example, \cite{CKP}.

In addition, $\phi\not\equiv 0$  on $\RR^d$ if $u\ne 0$ in $L^2({\mathbb S}^{d-1})$, and $|k|\ne0$.

The energy $E$ such that (\ref{eq:schr1}), (\ref{eq:schr2}), (\ref{eq:tr_eig3}) are fulfilled with non-trivial $\phi$, $\psi$ is called an
\textbf{interior transmission eigenvalue} for equation (\ref{eq:schr1}) in the domain $\cal D$; see \cite{CH2,CH,CKP,CM,Kir,NgNg,RS}.

\begin{remark}\label{rem:2.1} If $E\in\RR$, $E>0$, is a transmission eigenvalue of multiplicity $N$ in the strong sense for equation~(\ref{eq:1}), then $E$ is an interior transmission eigenvalue of multiplicity $\ge N$ for equation~(\ref{eq:1}) in any domain $\cal D$ as in (\ref{eq:supp}).

\end{remark}  

\section{Main results} 

Our results on transmission eigenvalues in the strong sense (in the sence of equation (\ref{eq:tr_eig})) for mutlipoint scatterers (\ref{eq:1.2}) are formulated in Theorem~\ref{thm:1} and in Proposition~\ref{prop:1}.
\label{sect:3}
\begin{theorem}\label{thm:1}
  Let $v$ be a multipoint scatterer (\ref{eq:1.2}) of the Bethe-Peierls-Fermi-Zeldovich-Beresin-Faddeev type in dimension $d=2$ or $d=3$.

  Then each energy $E>0$ is a transmission eigenvalue of infinite multiplicity for equation (\ref{eq:1}) in the strong sense.
\end{theorem}  
\begin{proposition}\label{prop:1}
  Let $v$ be a single point scatterer in dimension $d=1$ as in (\ref{eq:1.2}) with $n=1$.

  Then each energy $E>0$ is a transmission eigenvalue for equation (\ref{eq:1}) in the strong sense.
\end{proposition}  

To prove our results on interior transmission eigenvalues for  mutlipoint scatterers (\ref{eq:1.2}) we use the following lemma, which is also of independent interest.

\begin{lemma}\label{lem:1}
Let $v$ be a miltipoint scatterer as in (\ref{eq:1.2}), (\ref{eq:supp}),  $\phi$ satisfy (\ref{eq:schr2}), and $\phi(y_j)=0$ for all $j=1,\ldots,n$. Then $\phi$ also satisfies (\ref{eq:schr1}) . 
\end{lemma} 
In Lemma~\ref{lem:1} we assume that in equations  (\ref{eq:schr1}), (\ref{eq:schr2}) the energy $E\in\CC$.     

Our results on interior transmission eigenvalues (in the sense of (\ref{eq:schr1})--(\ref{eq:tr_eig3}))  for mutlipoint scatterers (\ref{eq:1.2}) are formulated in Theorem~\ref{thm:2} and in Proposition~\ref{prop:2}.

\begin{theorem}\label{thm:2}
  Let $v$ be a multipoint scatterer as in (\ref{eq:1.2}) of the Bethe-Peierls-Fermi-Zeldovich-Beresin-Faddeev type in dimension $d=2$ or $d=3$, satisfying condition (\ref{eq:supp}). 

  Then each energy $E\in\CC$ is an interior transmission eigenvalue of infinite multiplicity for equation (\ref{eq:1}) in the domain $\cal D$.
\end{theorem}  

\begin{proposition}\label{prop:2}
 Let $v$ be a single point scatterer in dimension $d=1$ as in (\ref{eq:1.2}) with $n=1$, satisfying condition (\ref{eq:supp}).

  Then each energy $E\in\CC$ is an interior transmission eigenvalue for equation (\ref{eq:1}) in the domain $\cal D$.
\end{proposition}

In Theorem~\ref{thm:2} and Proposition~\ref{prop:2}, $\cal D$ is the domain of formulas (\ref{eq:supp}), (\ref{eq:schr1})--(\ref{eq:tr_eig3}).

\section{Proofs of the main results}\label{sect:4}

\begin{proof}[Proof of Theorem~\ref{thm:1}]
  It is convenient to use that
  \begin{equation}\label{eq:rec}
    f(k,l) = f(-l,-k), \ \ \mbox{where} \ \ k,l\in\RR^d, \ k^2=l^2=E.
  \end{equation}
The reciprocity property (\ref{eq:rec}) follows from formula (\ref{eq:2.7}), equation (\ref{eq:2.10}) and the property that $A=A^t$. 

Due to (\ref{eq:2.7}), (\ref{eq:rec}), we have that 
\beq
\label{eq:4.2}
f(k,l) = \frac{1}{(2\pi)^d}  \sum\limits_{j=1}^n q_{j}(-l) 
e^{ik y_j}.
\eeq

Due to (\ref{eq:2.13}), (\ref{eq:4.2}), we have that equation (\ref{eq:tr_eig}) is fulfilled for each $u\in L^2({\mathbb S}^{d-1})$ such that
\beq
\label{eq:4.3}
\int_{{\mathbb S}^{d-1}} e^{i |k| \theta  y_j} u(\theta) d \theta =0, \ \ \mbox{for all} \ \ j=1,\ldots,n.
\eeq
One can see that a finite system of homogeneous linear equations (\ref{eq:4.3}) in $L^2({\mathbb S}^{d-1})$ has infinite-dimensional space of solutions for $d=2,3$.

Therefore $E$ is a transmission eigenvalue of infinite multiplicity in the strong sense for equation (\ref{eq:1}). 
\end{proof}

\begin{proof}[Proof of Proposition~\ref{prop:1}]
It is convenient to use that, for $d=1$, $n=1$:
\beq
\label{eq:4.4}
f(k,l) = \frac{1}{2\pi}   \frac{e^{ ik y_1}  e^{-il y_1}}{\alpha_1 + (2i |k|)^{-1}}, \ \ \mbox{where} \ \ k,l\in\RR, \ \ k^2=l^2=E.
\eeq
Due to (\ref{eq:4.4}), we have that equation (\ref{eq:tr_eig}) is fulfilled for each $u=(u^-,u^+)\in L^2({\mathbb S}^0)=\CC^2$ such that 
\beq
\label{eq:4.5}
 e^{-i |k| y_1}  u^- +  e^{i |k| y_1}  u^+ =0.
\eeq

One can see that equation  (\ref{eq:4.5}) has one-dimensional space of solutions.

Therefore $E$ is a transmission eigenvalue in the strong sense for equation (\ref{eq:1}).
\end{proof}

\begin{proof}[Proof of Lemma~\ref{lem:1}]
  This result follows from the definition of solutions of equation (\ref{eq:1}); see Introduction. Indeed, equation  (\ref{eq:1.1}) is fulfilled for $\psi=\phi$ in ${\cal D}$ containing all $y_j$, and we have that:
  \begin{enumerate}
  \item If $d=1$, then
    $$
    \phi(y_j)=0, \ \ \phi'(y_j+0)= \phi'(y_j-0),
    $$
    and condition (\ref{eq:1.4}) is fulfilled;
  \item If $d=2,3$, then  
    $$
    \phi_{j,-1}=\phi_{j,0}=0,
    $$
    and condition (\ref{eq:1.6}), for $d=2$, or condition  (\ref{eq:1.8}), for $d=3$, is fulfilled.
\end{enumerate}
\end{proof}

\begin{proof}[Proof of Theorem~\ref{thm:2}]
  For $E\in\RR$, $E>0$, the result follows from  Remark~\ref{rem:2.1} and Theorem~\ref{thm:1}.

  For $E\in\CC$, the proof is as follows. Let $\phi_l$, $l=1,\ldots,\infty$, denote an infinite set of linearly independent smooth solutions of (\ref{eq:schr2}) in $\cal D\cup\partial{\cal D}$. For any $N\in\ZZ$, $N>n$, consider the linear system on $z=(z_1,\ldots,z_N)\in\CC^N$:
\beq\label{eq:4.6}
 \sum\limits_{l=1}^N z_l \phi_l(y_j) =0.
\eeq
This system has at least $N-n$ dimensional space of solutions. For each $z=(z_1,\ldots,z_N)$ satisfying (\ref{eq:4.6}) the function
$$
\Phi(x) = \sum\limits_{l=1}^N z_l \phi_l(x)
$$
satisfies (\ref{eq:schr2}), and $\Phi(y_j)=0$ for $j=1,\ldots,n$. Therefore, by Lemma~\ref{lem:1}, $\Phi(x)$ satisfies  (\ref{eq:schr1}) and (\ref{eq:schr2}) simultaneously! In addition, $\Phi(x)\not\equiv 0$ in $\cal D$, due to linear independence of $\phi_l$. Thus, $E$ is an interior transmission eigenvalue for equation (\ref{eq:schr1}) in the sense of (\ref{eq:schr1})--(\ref{eq:tr_eig3}), where $\psi\equiv\phi\equiv\Phi$ in ${\cal D}\cup\partial{\cal D}$. In addition, the space of such $\Phi$ has dimension $\ge N-n$. Therefore, $E$ is an interior transmission eigenvalue of multiplicity, at least, $N-n$ for arbitrary large $N$, i.e. the multiplicity of $E$ is infinite. 
\end{proof}

\begin{proof}[Proof of Proposition~\ref{prop:2}]
 For $E\in\RR$, $E>0$, the result follows from  Remark~\ref{rem:2.1} and Proposition~\ref{prop:1}.

 For $E\in\CC$, the proof is as follows. Denote by $\Phi$ a non-zero solution of (\ref{eq:schr2}) such that $\Phi(y_1)=0$. By Lemma~\ref{lem:1}, $\Phi(x)$ satisfies  (\ref{eq:schr1}) and (\ref{eq:schr2}) simultaneously. Therefore, $E$ is an interior transmission eigenvalue for equation (\ref{eq:schr1}) in the sense of (\ref{eq:schr1})--(\ref{eq:tr_eig3}), where $\psi\equiv\phi\equiv\Phi$ in ${\cal D}\cup\partial{\cal D}$.
 
 \end{proof}

\end{document}